\documentclass[9pt,technote]{IEEEtran}

\usepackage{cite}
\usepackage{graphicx}

\usepackage[hidelinks]{hyperref}

\usepackage{tikz}
\usetikzlibrary{arrows,positioning,calc,fit,shapes.geometric}
\usetikzlibrary{arrows.meta,decorations.pathreplacing}
\tikzset{>={Stealth[length=1.5mm, width=1.5mm]}}

\usepackage{xcolor}

\usepackage{amsfonts}
\usepackage{mathtools}
\usepackage{accents}

\newcommand{\wt}[1]{\widetilde{#1}}
\newcommand{\mc}[1]{\mathcal{#1}}
\usepackage{mathrsfs}

\newcommand{\RR}{\mathbb{R}}
\newcommand{\NN}{\mathbb{N}}

\newcommand{\PP}{\mathbb{P}}
\newcommand{\EE}{\mathbb{E}}

\usepackage{bbm}

\newcommand{\conv}[1]{\mathrm{conv}\left(#1\right)}
\newcommand{\oconv}[1]{\overline{\mathrm{conv}}\left(#1\right)}

\newtheorem{theorem}{Theorem}
\newtheorem{proposition}{Proposition}

\newtheorem{lemma}{Lemma}



\title{Rate-Cost Tradeoffs in Nonlinear Control}
\author{Eray Unsal Atay, Venkat Chandrasekaran, Victoria Kostina%
\thanks{This work was supported in part by AFOSR grants \mbox{FA9550-23-1-0070} and FA9550-23-1-0204, by Caltech’s Center for Sensing to Intelligence, and by the Carver Mead New Adventures Fund.}%
\thanks{The authors are with the California Institute of Technology, Pasadena, CA 91106, USA.
{\tt\small \{eatay,venkatc,vkostina\}@caltech.edu}}%
}

\begin{document}

\maketitle

\begin{abstract}
    We study finite-horizon rate-cost tradeoffs for general stochastic control
    systems, including nonlinear systems. Under an average control cost
    constraint,
    we show that
    the minimum expected communication rate lies within an additive
    logarithmic gap above
    the normalized directed information minimized over causal control policies.
    This establishes directed information as the operationally relevant quantity governing rate-limited control,
    extending its previously established roles
    in causal source coding and linear quadratic Gaussian (LQG) control
    to general nonlinear control systems.
\end{abstract}

\section{Introduction}

In many applications, a control system is observed through its states
but communicates with the controller over a rate-limited link,
as shown in Fig.~\ref{fig:sys1}.
As a result, the controller receives an encoded bitstream describing the state rather than
observing the state directly
\cite{nair2007feedback}—examples include remote operation of drones and autonomous vehicles~\cite{HespanhaNaghshtabriziXu2007} or wireless robot control over congested networks~\cite{Oboe2003}, where the transmission of data from the system to the base station is rate-limited.
This setting leads to the \emph{rate-limited control} problem,
which couples source coding and control,
with performance measured by an average \emph{control cost} in terms of the states and actions.
The key question is: what minimum communication rate guarantees that the control cost stays below a target level?
This fundamental tradeoff is captured by the \emph{rate-cost function}~\cite{KostinaHassibi2018, NakahiraXiaoKostinaDoyle2018BiomolecularControl, KostinaHassibi2019RateCostControl}, the control-theoretic analogue of the rate-distortion function (RDF) in source coding, which characterizes the minimum communication rate needed to represent a source subject to a prescribed distortion constraint.

In
the classical rate-distortion problem,
one seeks to represent a random information source with as few bits as possible
while satisfying a fidelity constraint.
The central quantity is the RDF, introduced by Shannon, which characterizes the fundamental limit under a
given
distortion measure between source and reproduction symbols
\cite{Shannon1959}:
for a distortion level $D$,
the RDF
$R(D)$ is the minimum number of bits per symbol needed to describe the source subject to the distortion not exceeding $D$, in the limit of large coding blocklength.
Rate-distortion theory~\cite{Berger1971,CoverThomas2006} underpins the design and analysis of lossy compression systems for speech, image, and video~\cite{GershoGray1992};
scalar and vector quantization~\cite{GrayNeuhoff1998};
and, more broadly, communication and control systems
operating
under bit-rate constraints~\cite{
KostinaHassibi2018,TanakaEsfahaniMitter2018_LQGMinDI}.

\begin{figure}[t]
    \centering

    \begin{tikzpicture}[
        >=Latex,
        every node/.style={font=\footnotesize},
        node distance=10mm and 10mm,
        block/.style={
            draw, rounded corners, thick, align=center,
            minimum height=10mm, inner sep=4pt, fill=white
        },
        plant/.style={block, fill=gray!15, minimum width=2.0cm, minimum height=0.9cm},
        enc/.style={block, minimum width=2.0cm, minimum height=0.9cm},
        ctrl/.style={block, minimum width=2.0cm, minimum height=0.9cm},
        line/.style={-Latex, thick}
    ]

    \node[plant] (source) {\textbf{System}};
    \node[enc, right=of source] (encoder) {\textbf{Encoder}};
    \node[ctrl, right=of encoder] (controller) {\textbf{Controller}};

    \draw[line] (source.east) -- node[above] {$X_t$} (encoder.west);
    \draw[line] (encoder.east) -- node[above] {$B_t$} (controller.west);
    \draw[line]
          (controller.south) -- ++(0,-0.5) coordinate (ctrlbottom)
          -- (source |- ctrlbottom) coordinate[midway] (mid)
          -- (source.south);

    \node[below] at (mid) {$U_t$};

    \end{tikzpicture}
    \caption{Control system with a rate-limited link to the controller.
    $X_t$, $B_t$, and $U_t$ denote the system state, transmitted message, and control action at time $t$, respectively.
    The shaded block
    is fixed and known; the white blocks
    are design components.}
    \label{fig:sys1}
\end{figure}

Despite this analogy, rate-limited control differs fundamentally from classical source coding:
in control systems, each system state
must be encoded and reproduced with \emph{zero delay}, i.e., before the next state is observed,
so that the controller can act immediately.
The source-coding counterpart of this zero-delay requirement is \emph{sequential (causal) source coding}~\cite{NeuhoffGilbert1982},
in which each reproduction symbol must be generated before the next source symbol is observed. Such real-time constraints arise naturally in applications such as low-latency broadcasting and remote sensing and control~\cite{Tatikonda2000}.
This \emph{causal}
(also called \emph{sequential} or \emph{zero-delay})
constraint contrasts with classical source coding, where the encoder can access the entire block before producing a description.
Zero-delay eliminates lookahead,
changes the structure of optimal encoders and decoders,
and typically increases the minimum achievable rate
\cite{
GorbunovPinsker1973
}.

This viewpoint
naturally
suggests directed information as the quantity that governs rate-limited control as well.
In the linear quadratic Gaussian (LQG) setting, prior work derives explicit lower bounds and near-tight achievable schemes for the rate-cost function by leveraging the separation between communication and control
available in that setting~\cite{KostinaHassibi2019RateCostControl, TatikondaSahaiMitter2004_StochasticLinearControl
, TanakaEsfahaniMitter2018_LQGMinDI
}.
For nonlinear systems, however, where such a separation principle need not hold in general,
existing work studies whether limited communication is enough to keep the system stable~\cite{YukselBasar2013_SNC,ZhengEtAl2018_HowMuchInfo}
and derives specialized upper bounds for restricted system classes such as feedforward or strict-feedback nonlinear systems~\cite{JiangLiu2013_QuantizedNonlinearSurvey}.
Accordingly, unlike uncontrolled causal source coding and LQG control—where directed information is known to govern the
rate-cost tradeoffs—no prior work
for general nonlinear systems provides information-theoretic lower and upper bounds on the operational rate-cost function, with both bounds being only in terms of the minimum directed information.

In this paper, we
demonstrate that the operational rate-cost tradeoff of general nonlinear control systems is governed by a directed-information minimization.
We derive a lower bound on the communication rate necessary for a prescribed control cost, extending previously known methods to the nonlinear setting.
Our main result is a finite-horizon upper bound on this
rate, derived using the strong functional representation lemma (SFRL)~\cite{LiElGamal2018SFRL,Li2025tighterSFRL}.
These results also yield converse and achievability bounds on the communication rate required to keep the average cost bounded independently of the horizon. Under a quadratic lower bound on the stage cost, such cost bounds imply bounded time-averaged state second moments, providing a connection to stability.
We
specialize our framework to LQG control and sequential source coding, two
canonical
special cases of our setting.
We also provide numerical illustrations for linear and quadratic plants.


\subsection{Information measures} \label{subsec:information-measures}

For a positive integer $n$,
we denote the set $\{1,\ldots,n\}$ by $[n]$
and write ${X_{[n]} = (X_1 , \ldots , X_n)}$.
All information quantities are measured in bits, and all logarithms are in base $2$.

The mutual information of random variables $X,Y$
equals
\begin{equation}
    I(X;Y)
    \coloneq
    \mathbb{E}_{P_{Y|X}P_X}
    \! \left[
    \log
    \frac{dP_{Y|X}(Y|X)}{dP_Y(Y)}
    \right] ,
\end{equation}
and quantifies the information $X$ and $Y$ contain about each other.
The conditional mutual information of $X,Y$ given $Z$ is
\begin{equation}
    I(X;Y \, | \, Z)
    \coloneq
    \mathbb{E}_{P_{Y|X,Z}P_{X,Z}}
    \!\left[
    \log
    \frac{dP_{Y|X,Z}(Y|X,Z)}{dP_{Y|Z}(Y|Z)}
    \right] ,
\end{equation}
and quantifies the information shared by $X$ and $Y$ beyond what is
known from $Z$.
The directed information $I \! \left( X_{[n]} \to U_{[n]} \right)$ from the state sequence $X_{[n]}$ to the control
sequence $U_{[n]}$
is given by
\cite{Massey1990}
\begin{equation} \label{eq:directed-information-control}
    I \! \left( X_{[n]} \to U_{[n]} \right) \coloneq \sum_{t=1}^n I \! \left( X_{[t]} ; U_t \; \big| \; U_{[t-1]} \right) .
\end{equation}
Each
conditional mutual information term $I \! \left( X_{[t]} ; U_t \big| U_{[t-1]} \right)$ measures the new information about the state history $X_{[t]}$ contained in the current control action $U_t$, beyond what is already conveyed by past control actions $U_{[t-1]}$.
Summing these terms captures the causal information flow from the state to the control sequence.

\subsection{System model and problem formulation} \label{subsec:system-model}

Consider the setting in
Fig.~\ref{fig:sys1}:
a control system is observed through its states $X_t \in \mc X$.
The encoder produces messages $B_t$, and the controller applies actions $U_t \in \mc U$ that influence future states.
We assume that the state and action alphabets $\mc X, \mc U$ are standard Borel alphabets\footnote{We assume the alphabets to be standard Borel
to apply
the SFRL bound,
which holds under this assumption~\cite{LiElGamal2018SFRL}.},
and that each message $B_t$ is a finite binary string, i.e., $B_t \in \mc B_t(B_{[t-1]}) \subset \{0,1\}^\star$,
where each $\mc B_t(b_{[t-1]})$ is a prefix-free codebook of binary strings
for the fixed past message history $b_{[t-1]}$.

Fix a time horizon $n\ge 1$.
The state kernel
\begin{equation}
    P_{X_{[n]} \parallel U_{[n-1]}}
    = \prod_{t=1}^n P_{X_t \mid X_{[t-1]}, U_{[t-1]}}
\end{equation}
is fixed, whereas the encoder and controller kernels
\begin{align}
    P_{B_{[n]} \parallel X_{[n]}} &= \prod_{t=1}^n P_{B_t \mid X_{[t]}, B_{[t-1]}} , \label{eq:encoding-kernel} \\
    P_{U_{[n]} \parallel B_{[n]}} &= \prod_{t=1}^n P_{U_t \mid B_{[t]}, U_{[t-1]}} \label{eq:control-kernel}
\end{align}
are to be designed.
Here, to emphasize that the considered stochastic kernels are nonanticipatory,
we use the causal conditioning notation introduced by Kramer~\cite{Kramer1998}.
We refer to the pair $\bigl( P_{B_{[n]} \parallel X_{[n]}},\, P_{U_{[n]} \parallel B_{[n]}} \bigr)$
as the \emph{encoding-and-control policy}.
These design kernels,
together
with the
system dynamics $P_{X_{[n]} \parallel U_{[n-1]}}$,
induce a joint law on $\left( X_{[n]},B_{[n]},U_{[n]} \right)$.

The communication rate is quantified by the expected number of transmitted bits per time step,
given by
\begin{equation} \label{eq:rate}
    \frac{1}{n}\sum_{t=1}^n \EE[\ell(B_t)] ,
\end{equation}
where $\ell \colon \{0,1\}^\star \to \NN$ outputs the length of the binary string in its argument.
Each action $U_t$ incurs a
per-stage
control cost $c(X_t,U_t)$,
where
$c \colon \mc X \times \mc U \to \RR_{\ge0}$ is measurable.
Control performance is measured by the average control cost
\begin{equation} \label{eq:cost}
    \frac{1}{n}\sum_{t=1}^n \EE \! \left[c(X_t,U_t)\right].
\end{equation}
Our goal is to characterize the
\emph{rate-cost tradeoff},
quantified by the rate-cost function $R_n(D)$
which denotes the minimum communication rate
required to achieve
an average control cost of at most $D$.
Formally,
\begin{align} \label{eq:RnD-defn-first}
    R_n(D)
    \coloneq
    \inf \bigl\{ \eqref{eq:rate}
    \, \colon \,
    \exists \text{ design kernels s.t. }
    \eqref{eq:cost}
    \le D \bigr\} .
\end{align}


\subsection{Contributions}

Our main result connects the operational rate-cost function with the following information-theoretic quantity:
\begin{equation} \label{eq:def-FnD}
        F_n(D)
        \!\coloneq\!
        \inf \Bigl\{
        \frac{1}{n} I \! \left( X_{[n]} \to U_{[n]} \right)
        \colon
        \exists \, P_{U_{[n]} \parallel X_{[n]}} \,\text{s.t.}\, \eqref{eq:cost} \!\le\! D
        \Bigr\}.
\end{equation}
Here, the optimization objective is the directed information,
which is a function of the system dynamics
$P_{X_{[n]} \parallel U_{[n-1]}}$
and the state-to-control kernel
${P_{U_{[n]} \parallel X_{[n]}}
= \prod_{t=1}^n P_{U_t \mid X_{[t]}, U_{[t-1]}}}$.
Note that,
in the system shown in
Fig.~\ref{fig:sys1},
the kernel $P_{U_{[n]} \parallel X_{[n]}}$ is induced
by the encoding-and-control policy
$\bigl( P_{B_{[n]} \parallel X_{[n]}},\, P_{U_{[n]} \parallel B_{[n]}} \bigr)$.
In the optimization~\eqref{eq:def-FnD},
the finite-bit-rate constraint is relaxed,
and the optimization is
instead
over all state-to-control kernels
satisfying the cost constraint.
Moreover,
the
optimization~\eqref{eq:def-FnD}
is over the entire
state-to-control kernel $P_{U_{[n]} \parallel X_{[n]}}$
across all $n$ stages, so $F_n(D)$ captures a full-horizon optimization rather than a
myopic per-stage one; a policy that is optimal locally at each time
step
may not be optimal over the full horizon.
Note that
for fixed $P_{X_{[n]} \parallel U_{[n-1]}}$,
the map
${P_{U_{[n]} \parallel X_{[n]}} \mapsto
I \! \left( X_{[n]} \to U_{[n]} \right)}$
is convex~\cite[Thm.\ 6]{Charalambous2016DirectedInformation},
and the cost constraint
$\frac{1}{n} \sum_{t=1}^n \EE\bigl[c(X_t,U_t)\bigr] \le D$
is linear in this kernel.
As a result, the feasible set is convex, and the minimization
in~\eqref{eq:def-FnD} is a convex optimization problem.

Our main result
bounds the difficult optimization~\eqref{eq:RnD-defn-first}
in terms of the
information-theoretic
optimization~\eqref{eq:def-FnD}:
\begin{equation} \label{eq:Rn-Fn-0}
    F_n(D)
    \,\le\,
    R_n(D)
    \,\le\,
    F_n(D)
    +
    \log\bigl(F_n(D)+3.4\bigr)
    +
    2
    +
    \frac{1}{n} ,
\end{equation}
demonstrating that the minimum directed information $F_n(D)$
governs the rate-cost tradeoff
up to an additive gap of at most ${\log\bigl(F_n(D)+3.4\bigr)+2+\frac1n}$.

The converse and achievability bounds in~\eqref{eq:Rn-Fn-0}
characterize the
fundamental
rate-cost tradeoff:
no scheme can operate below rate $F_n(D)$, while rates within
an additive logarithmic gap of $F_n(D)$ are achievable.
The converse
$R_n(D) \ge F_n(D)$
was
known for LQG control~\cite{TanakaEsfahaniMitter2018_LQGMinDI},
and
the same proof strategy extends to the general nonlinear setting considered here.
Our primary contribution is the nonasymptotic achievability bound
given by the upper bound in~\eqref{eq:Rn-Fn-0}.
The proof uses a construction that applies the conditional
SFRL
stage by stage,
and does not rely on the communication--control separation available in LQG systems.
To the best of our knowledge, this is the first use of the SFRL to derive an achievability bound in a sequential, controlled setting.
We also establish a connection to stability through horizon-independent bounds on average cost and, in the LQG specialization, recover the scalar and vector data-rate converses.

Informational characterizations of this kind are important
because they provide benchmarks
that can guide the design of practical communication and control algorithms.
Relating the operational definition~\eqref{eq:RnD-defn-first}
to the analytic expression~\eqref{eq:def-FnD} facilitates
analysis and, when the directed-information minimization is tractable,
enables numerical computation.

Along any sequence of problem instances for which $F_n(D)$ is finite and tends to infinity,
the logarithmic gap is negligible relative to the leading term, and ${R_n(D) \big/ F_n(D)\to1}$.
To our knowledge, this is the first nonasymptotic achievability bound for the general class of control systems considered here, including nonlinear systems,
that is asymptotically tight in the first-order sense in such high-information regimes.
Such regimes arise, for example, in high-resolution causal coding of Gauss--Markov sources as the distortion threshold approaches zero, ${D\downarrow0}$~\cite{DerpichOstergaard2012},
and in LQG control as the cost threshold approaches the full-information minimum, $D\downarrow D_{\min}$~\cite{KostinaHassibi2019RateCostControl}.
They also arise in high-dimensional causal source-coding settings,
where the spatial dimension of the source increases
via a growing number of independent components~\cite{stavrou2018zerodelay,kostina2021ceo}.


\subsection{Related work}

\paragraph*{Sequential source coding}
In
sequential (zero-delay) source coding,
the source
$\{X_t\}_{t\ge1}$
is
characterized by kernels $P_{X_t \mid X_{[t-1]}}$,
and
one designs causal mappings to produce
reproduction symbols $\hat X_t, \, t \in [n]$,
under a distortion constraint between $X_{[n]}$ and $\hat X_{[n]}$.
The key quantity is the normalized directed information ${\frac{1}{n} I \bigl( X_{[n]} \to \hat X_{[n]} \bigr)}$~\cite{Massey1990}, whose minimization
over causal reproduction kernels
defines the sequential (causal) rate-distortion function
\begin{equation} \label{eq:directed-info}
    F_n^{\mathrm{sc}}(D)
    \, = \,
    \inf_{ \substack{ P_{\hat X_{[n]} \parallel X_{[n]}} \colon \\ \EE \left[d \left( X_{[n]}, \hat X_{[n]} \right) \right] \le D } } \frac{1}{n} I \! \left( X_{[n]} \to \hat X_{[n]} \right) ,
\end{equation}
where $d \colon \mc X^n \times \hat{\mc X}^n \to \RR_{\ge 0}$
is a distortion measure quantifying
reproduction quality.
The structure of the directed information expression in~\eqref{eq:directed-info} was implicit in finite-horizon results:~\cite[Thm.~1]{ViswanathanBerger2000} treated the case of two time steps, and~\cite[Thm.~1]{Yang2011} extended
the result
to an arbitrary horizon $n$. A more general setting with $K$ observers that only indirectly observe the source was analyzed in~\cite{kostina2021ceo}; when specialized to $K=1$ with noiseless observations,
it reduces to sequential source coding.
In general, the quantity in~\eqref{eq:directed-info} need not equal the operational causal rate-distortion function,
that is,
the minimum achievable coding rate under the distortion constraint,
as demonstrated by
the counterexamples in~\cite{BorkarMST05, JohnstonMP14}.
Operational tightness is recovered in the ``large spatial dimension'' regime where each state is a vector with $k$ i.i.d.\ coordinates: for fixed $n$, letting $k\to\infty$ yields an operationally tight directed-information characterization~\cite{Tatikonda2000,Yang2011,Ma2011,kostina2021ceo
}.
This line of work focuses on coding of uncontrolled sources and does not treat control systems.

\paragraph*{Rate-limited control}
The study of control under communication constraints asks how much information must be exchanged between sensors and controllers to achieve a target closed-loop performance.
An information-theoretic viewpoint is to regard the long-term control cost as a distortion level and to define a rate-cost function~\cite{KostinaHassibi2018} that characterizes the minimum communication rate needed to keep that cost below a prescribed value; this perspective enables one to apply rate-distortion theory to control problems~\cite{KostinaHassibi2018, KostinaHassibi2019RateCostControl, NakahiraXiaoKostinaDoyle2018BiomolecularControl}.
Kostina and Hassibi~\cite{KostinaHassibi2018, KostinaHassibi2019RateCostControl} develop a directed-information characterization for the minimum rate needed to achieve a target quadratic cost in linear Gaussian control over communication links, along with lower bounds and near-optimal coding--control schemes.
Earlier work by Tatikonda et al.~\cite{TatikondaSahaiMitter2004_StochasticLinearControl} analyzes LQG control over noisy and noiseless channels, gives conditions where estimation and control can be separated, and relates the performance gap to a sequential rate-distortion term.
Tanaka et al.~\cite{TanakaEsfahaniMitter2018_LQGMinDI} consider minimization of directed information under a quadratic cost constraint and obtain the optimal linear Gaussian controller through semidefinite programming.
Overall, these works focus on LQG control—a much more limited setting than general nonlinear control systems, which remain underexplored.

\paragraph*{Functional representation lemmas}
The (standard) functional representation lemma (FRL) asserts that for any pair of random variables $(X,Y)$, there exists a random variable $Z \perp X$ and a measurable function $g$ such that $Y = g(X,Z)$.
The \emph{strong} functional representation lemma (SFRL)~\cite{LiElGamal2018SFRL}
states that for any pair $(X,Y)$ with $I(X;Y) < \infty$,
there exist a random variable $Z \perp X$ and a measurable function $g$ such that
$Y = g(X,Z)$, and
$H(Y\mid Z) \le I(X;Y)+\log\bigl(I(X;Y)+1\bigr)+4$.
This upper bound was later tightened to $I(X;Y)+\log\bigl(I(X;Y)+3.4\bigr)+1$~\cite{Li2025tighterSFRL}.
The SFRL can be used to turn mutual-information expressions into achievability bounds~\cite{LiElGamal2018SFRL}.

\section{Nonasymptotic Rate-Cost Bounds} \label{sec:results}

\subsection{Converse and achievability bounds} \label{subsec:bounds}

We first
state
a converse bound in Proposition~\ref{prop:conv} below.
For the special case of LQG control, this bound was established in~\cite{TanakaEsfahaniMitter2018_LQGMinDI},
and a similar argument applies to nonlinear control.

\begin{proposition}[Converse] \label{prop:conv}
    Fix
    $n\ge 1$ and a cost level
    ${D\ge 0}$.
    The rate-cost function $R_n(D)$ is lower-bounded as
    \begin{equation} \label{eq:conv}
        R_n(D) \ge F_n(D) .
    \end{equation}
\end{proposition}
The proof
is given in Section~\ref{subsec:proof-prop:conv}.

Our main result is an achievability theorem, stated as Theorem~\ref{thm:ach-sfrl} below.
It states that the rate-cost function $R_n(D)$ lies within an additive logarithmic gap of $F_n(D)$.

The main technical ingredient is the strong functional representation lemma (SFRL)~\cite{LiElGamal2018SFRL, Li2025tighterSFRL}.
The conditional form we use is stated below as a lemma.

\begin{lemma}[{Conditional SFRL \cite[Thm.\ 14]{Li2025tighterSFRL}}] \label{lem:cond-sfrl}
    Let $X,Y,U$ be random variables with ${I(X;Y\mid U)<\infty}$.
    Then there exist a random variable $Z$ and a measurable function $g$ such that:
    \begin{itemize}
        \item[(i)] $Z \perp (X,U)$,
        \item[(ii)] $Y = g(X,U,Z)$ 
                and, given $(U,Z)$, the random variable $Y$ is discrete-valued,
        \item[(iii)] $H(Y\,|\,U,Z) \le I(X;Y\,|\,U) + \log\bigl(I(X;Y\,|\,U)+3.4\bigr)+1$.
    \end{itemize}
\end{lemma}

\begin{theorem}[Achievability] \label{thm:ach-sfrl}
    $R_n(D)$ is upper-bounded as
    \begin{equation} \label{eq:ach-sfrl}
        R_n(D)
        \ \le \
        F_n(D)
        \,+\,
        \log\bigl(F_n(D)+3.4\bigr)
        \,+\,
        2
        \,+\,
        \frac{1}{n} .
    \end{equation}
\end{theorem}
The proof is given in Section~\ref{subsec:proof-thm:ach-sfrl}.

The
achievability guarantee
in Theorem~\ref{thm:ach-sfrl}
is
information-theoretic
and does not directly establish physical realization:
actual
implementation of the SFRL functions and codebooks
is beyond the scope of this work.

Proposition~\ref{prop:conv} and Theorem~\ref{thm:ach-sfrl}
together
yield
\eqref{eq:Rn-Fn-0},
i.e.,
$R_n(D)$
lies within an additive logarithmic gap of $F_n(D)$.
For the target cost $D$, a communication rate below the lower endpoint in~\eqref{eq:Rn-Fn-0} is impossible, a rate strictly above its upper endpoint is guaranteed achievable, and a rate in the intervening interval is not resolved by the present bounds.

\subsection{Connection to stability}

We next consider whether a fixed communication rate can keep
the average cost bounded independently of the horizon.
Under a quadratic lower bound on the stage cost, this also
yields a bound on the time-averaged state second moment,
providing a connection to stability.

For a fixed system, define
\begin{equation} \label{eq:def-F-stab}
    \overline F(D)\coloneq\limsup_{n\to\infty}F_n(D),
    \quad
    F_{\mathrm{stab}}\coloneq\inf_{0\le D<\infty}\overline F(D).
\end{equation}
The following proposition gives
converse and achievability bounds on the communication rate
required to keep the average cost bounded independently of the horizon.

\begin{proposition}[Horizon-independent bounds] \label{prop:stability}
    \hfill {}
    
    \begin{itemize}
        \item \emph{Converse:}
        If, for every sufficiently large $n$, there is a horizon-$n$
        encoding-and-control policy with average rate at most $R$
        and average cost at most $D$, where $R,D \!<\! \infty$ are independent
        of $n$,
        \begin{equation}
            R \ge \overline F(D) \ge F_{\mathrm{stab}}.
        \end{equation}

        \item \emph{Achievability:}
        If $F_{\mathrm{stab}}<\infty$, every rate satisfying
        \begin{equation} \label{eq:stab-ach}
            R > F_{\mathrm{stab}}
                + \log\bigl(F_{\mathrm{stab}}+3.4\bigr)
                + 2
        \end{equation}
        admits, for some fixed $D<\infty$ and every sufficiently large
        $n$, a horizon-$n$ encoding-and-control policy with average
        rate at most $R$ and average cost at most $D$.
    \end{itemize}
\end{proposition}
The proof
is
given
in
Section~\ref{subsec:proof-stability}.

If
the cost function satisfies
$c(x,u)\ge\beta\|x\|^2$ for some ${\beta>0}$,
any encoding-and-control policy with average cost at most $D$
also satisfies
\begin{equation} \label{eq:stab-state-bound}
    \frac1n\sum_{t=1}^n \EE \! \left[ \|X_t\|^2 \right]
     \,\le\,  \frac{1}{\beta n}\sum_{t=1}^n \EE \! \left[c(X_t,U_t)\right]
     \,\le\,  \frac D \beta
\end{equation}
for all sufficiently large $n$,
giving a
time-averaged
mean-square state bound
independent of
the horizon
$n$.
Without further assumptions relating state boundedness to the stage cost, however, the converse holds for bounded average cost, but does not necessarily apply to state boundedness.

\section{Special Cases and Examples} \label{sec:examples}

We specialize the general rate-cost bounds to two canonical settings
and provide numerical illustrations. 

\subsection{Linear quadratic Gaussian control}

In
LQG
control,
the system state
$X_t \in \RR^k$ evolves as
\begin{equation} \label{eq:LQG-1}
    X_{t+1} = A X_t + B U_t + W_t, \quad W_t \sim \mc N(0,\Sigma),
\end{equation}
where $U_t \in \RR^m$ is the control input, $W_t$ is the noise, and ${A \in \RR^{k \times k}}$ and ${B \in \RR^{k \times m}}$ are constant matrices,
and
\begin{equation} \label{eq:LQG-2}
    c(X_t, U_t) = X_t^T Q X_t + U_t^T M U_t, \quad Q,M \succeq 0.
\end{equation}
The relevant quantity governing the rate-cost tradeoff in LQG control is $F_n(D)$ given in~\eqref{eq:def-FnD}~\cite{TanakaEsfahaniMitter2018_LQGMinDI, KostinaHassibi2019RateCostControl},
and in that case it has a characterization more explicit than~\eqref{eq:def-FnD}.
It is shown in~\cite{TanakaEsfahaniMitter2018_LQGMinDI} that the
minimum directed information subject to the prescribed LQG cost constraint is attained by a linear Gaussian policy consisting of a virtual sensor, a Kalman filter, and a certainty-equivalent controller.
For
the
finite-horizon
problem,
this minimum
can be computed via
a finite-dimensional semidefinite program (SDP);
for the time-invariant infinite-horizon problem, the
corresponding
minimum admits a single-letter SDP characterization~\cite{TanakaEsfahaniMitter2018_LQGMinDI}.

In general, $F_n(D)$ does not equal the minimum communication rate required to achieve average cost $D$~\cite{KostinaHassibi2019RateCostControl}: it
is a lower bound,
but this bound is not known to be achievable exactly.
Known constructions asymptotically achieve cost $D$ at a rate of at most
$F(D) + O(\log k)$, where $k$ denotes the system state dimension~\cite{KostinaHassibi2019RateCostControl}.
Theorem~\ref{thm:ach-sfrl} gives the nonasymptotic achievability bound~\eqref{eq:ach-sfrl} for each horizon $n$.

Achievability bounds for LQG control
using
directed information have been established; in particular,
in the infinite-horizon setting,
the dithered-quantization
scheme of~\cite[Thm.~3]{Tanaka2016RatePrefixFreeLQG} achieves
a gap of
$\approx 0.7546\rho + 1$
bits per time step above the minimum directed information,
with $\rho\le k$
the rank of the optimal signal-to-noise
ratio matrix.
For zero-delay vector Gaussian source coding
under mean-square distortion,\linebreak
\cite[Thm.~5]{stavrou2018zerodelay}
obtains an asymptotic gap of
$\approx 0.2546\rho_s+1$
bits per vector above the nonanticipative rate--distortion
function, where $\rho_s$ counts active source dimensions.

Here, for all $n\ge1$
and all $D$ such that $F_n(D)<\infty$,
Theorem~\ref{thm:ach-sfrl} guarantees
an
encoding-and-control policy
with average cost at most $D$
and average expected rate at most
a value arbitrarily close to
${F_n(D)+\log\bigl(F_n(D)+3.4\bigr)+2+1/n}$
from above.
When $F_n(D)\to F(D)$, our overhead tends to
$\log\bigl(F(D)+3.4\bigr)+2$, which is at least
$\log(3.4)+2\approx3.766$ bits, exceeding the approximately
$1.755$-bit guarantee
of~\cite{Tanaka2016RatePrefixFreeLQG} in the scalar case.
Thus, this example illustrates an explicit finite-horizon
operational guarantee,
although specialized scalar LQG results are tighter.
Our main contribution is the general
achievability construction, which applies beyond
linear--Gaussian systems, including nonlinear systems.

We next specialize Proposition~\ref{prop:stability} to LQG systems,
obtaining horizon-independent
bounds,
with the converse also extending to bounded time-averaged state second moments.

\begin{proposition}[LQG horizon-independent bounds] \label{prop:lqg-stability}
    For the LQG system above, assume that $(A,B)$ is controllable,
    $Q,M,\Sigma\succ0$, and
    $X_1\sim\mc N(0,\Sigma_0)$ with $\Sigma_0\succ0$.
    Then
    \begin{equation} \label{eq:vector-Fstab}
        F_{\mathrm{stab}}=h(A)
        \coloneq\sum_{i:\,|\lambda_i(A)|>1}\log|\lambda_i(A)| ,
    \end{equation}
    and the following statements hold.

    \begin{itemize}
        \item \emph{Converse:}
        Any family of horizon-$n$ encoding-and-control policies with
        average rate
        at most $R$
        and a horizon-independent bound on
        either the average cost or
        $\frac1n \sum_{t=1}^n \EE[\|X_t\|^2]$, for all
        sufficiently
        large $n$,
        must satisfy
        $R \ge h(A)$.

        \item \emph{Achievability:}
        Every rate satisfying
        \begin{equation} \label{eq:vector-state-ach}
            R>h(A)+\log\bigl(h(A)+3.4\bigr)+2
        \end{equation}
        admits, for some fixed $D<\infty$ and every sufficiently large
        $n$, a horizon-$n$ encoding-and-control policy with average
        rate at most $R$, average cost at most $D$, and
        ${\frac1n\sum_{t=1}^n\EE[\|X_t\|^2]
        \le D/\lambda_{\min}(Q)}$.
    \end{itemize}
\end{proposition}

The proof
is
given
in
Section~\ref{subsec:proof-vector-lqg}.

For ${k=1}$, $h(A)=[\log|a|]^+$.
The sufficient rate in~\eqref{eq:vector-state-ach} is not sharp:
for every $R>h(A)$, there exists a single infinite-horizon policy
with asymptotic average rate below $R$
that satisfies $\sup_{t\ge1}\EE[\|X_t\|^2]<\infty$~\cite[Thm.~2.1]{NairEvans2004}.

\subsection{Sequential (causal) source coding}

We consider the special case $P_{X_{[n]} \parallel U_{[n-1]}} = P_{X_{[n]}}$,
i.e., the process $\{X_t\}_{t=1}^n$ is an (uncontrolled) source.
If we set ${U_t = \hat X_t}$ to be the reconstruction symbols and take the stage cost to be the distortion $c(X_t,U_t)=d(X_t,\hat X_t)$, then we recover \emph{sequential (causal) source coding}~\cite{NeuhoffGilbert1982,ViswanathanBerger2000}, in which each symbol must be encoded and decoded before the next source sample is observed. In this case, the quantity $F_n(D)$ 
is the \emph{causal rate-distortion function}~\cite{TatikondaSahaiMitter2004_StochasticLinearControl,Charalambous2014NRDF}.

Theorem~\ref{thm:ach-sfrl} implies that the minimum average
rate required to achieve average distortion $D$ is within an additive logarithmic gap of $F_n(D)$. In particular, the construction in the proof of Theorem~\ref{thm:ach-sfrl} yields a
causal
policy with rate arbitrarily close to the right-hand side of~\eqref{eq:ach-sfrl} for each
horizon
$n$.

In general, $F_n(D)$ does not coincide with the true minimum
achievable
average rate
under distortion $D$~\cite{BorkarMST05,JohnstonMP14}.
The gap vanishes in the large-spatial-dimension regime; that is, when $X_t$ is a $k$-dimensional source vector with i.i.d.\ coordinates, $n$ is fixed, and $k\to\infty$~\cite{Tatikonda2000,Yang2011,Ma2011,kostina2021ceo
}.
This celebrated result is a straightforward consequence of Theorem~\ref{thm:ach-sfrl}, as follows.
For $k$-dimensional source vectors $X_t$ with i.i.d.\ coordinates,
let $F_n^{(k)}(D) = k \, \wt F_n(D)$ denote the minimum directed information for a $k$-dimensional source under an average per-coordinate distortion constraint $D$, where $\wt F_n(D)$ denotes the minimum per-coordinate directed information.
Then, for fixed $n$ and $k\to\infty$, the minimum achievable per-coordinate rate is asymptotically characterized by the normalized quantity $\wt F_n(D)$,
because applying~\eqref{eq:ach-sfrl} to obtain an
upper bound on the per-coordinate rate yields
\begin{equation}
    \wt F_n(D) + \frac{\log \bigl( k \, \wt F_n(D)+3.4 \bigr)}{k} + \frac{2}{k} + \frac{1}{kn}
    \xrightarrow[k\to\infty]{}
    \wt F_n(D),
\end{equation}
so the lower and upper bounds in~\eqref{eq:Rn-Fn-0} coincide asymptotically,
showing that $\wt F_n(D)$ exactly characterizes the minimum achievable per-coordinate rate.

\subsection{Numerical illustrations} \label{subsec:numerical}

We first introduce two scalar plants:
\begin{align}
    X_{t+1} &= aX_t+U_t+W_t
    \quad \text{(linear)}, \label{eq:num-linear} \\
    X_{t+1} &= aX_t^2+U_t+W_t
    \quad \text{(quadratic)}, \label{eq:num-quadratic}
\end{align}
with $a = 0.6$ and the common stage cost ${c(x,u)=x^2+0.1u^2}$.
The initial state is uniform on $\mc X_1=\{-2,-1.5,\ldots,2\}$,
the action alphabet is $\mc U=\mc X_1$, and
the noise variables
$W_t$ are independent
and uniformly distributed on $\{-1,-0.5,0,0.5,1\}$.

For
numerical illustrations,
we consider
$k$-coordinate vector systems, with states
${\mathbf{X}_t=(X_{t,1},\ldots,X_{t,k})}$,
actions
${\mathbf{U}_t=(U_{t,1},\ldots,U_{t,k})}$, and noises
$\mathbf{W}_t=(W_{t,1},\ldots,W_{t,k})$,
where each coordinate is
an independent copy of the
corresponding
scalar model.
We use a short time horizon $n=4$, dimension $k=50$, and the
total stage cost
${c^{(k)}(\mathbf{x},\mathbf{u})=\sum_{i=1}^k c(x_i,u_i)}$.
Let $F_4^{(k)}(D)$ and $R_4^{(k)}(D)$ denote the corresponding
informational and operational rates, measured in total bits per
stage; thus, $D$ is the average total stage cost across all $k$
coordinates,
and ${F_4^{(k)}(D)=k\,F_4^{(1)}(D/k)}$,
where $F_4^{(1)}(D)$ denotes the single-coordinate quantity.

The short horizon and small
alphabets
keep the numerical optimizations manageable, since the number of
state--action histories grows rapidly with the horizon and
alphabet sizes.
Enlarging the alphabets leads to prohibitive
computation times in our implementation.
The independent-copy construction allows us to illustrate the
relative tightness of the bounds
without enlarging the scalar optimization.

Fig.~\ref{fig:numerical-bounds} displays the lower and upper
bounds in~\eqref{eq:Rn-Fn-0} for both $50$-dimensional systems.
For each plant, the logarithmic gap is modest relative to the
lower bound at larger values of $F_4^{(50)}(D)$, illustrating
the relative tightness of the bounds.

The two lower-bound curves also compare $F_4^{(50)}(D)$ for
the linear and quadratic vector plants.
The quadratic plant requires more directed information at a
common total cost over the shared plotted cost range, although
the two lower-bound curves are close near $D\approx45$.

\begin{figure}[tb]
    \centering
    \includegraphics[width=\columnwidth]{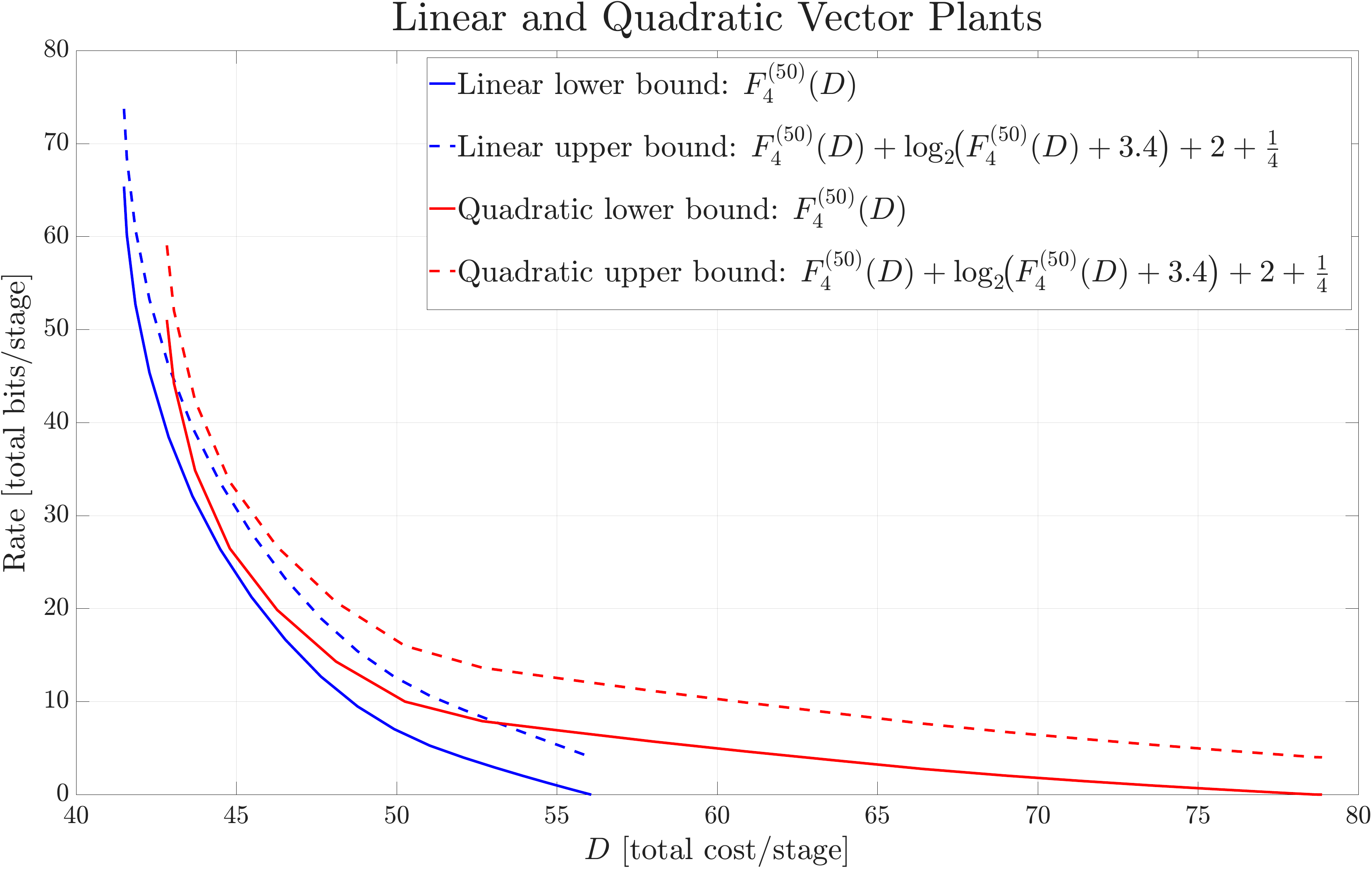}
    \caption{Lower and upper bounds in~\eqref{eq:Rn-Fn-0}
    for $R_4^{(50)}(D)$ of the linear and quadratic vector
    plants with horizon $n=4$ and dimension $k=50$.
    For each plant, the lower-bound curve is $F_4^{(50)}(D)$.
    The cost $D$ is the average total stage cost across all
    $50$ coordinates, and the rate is
    in total
    bits per stage.}
    \label{fig:numerical-bounds}
\end{figure}

\section{Proofs} \label{sec:proofs}

\subsection{Proof of Proposition~\ref{prop:conv}} \label{subsec:proof-prop:conv}

For each $t$, the controller selects $U_t$ using
$\left( B_{[t]},U_{[t-1]} \right)$; equivalently,
$U_t \,-\, \left( B_{[t]},U_{[t-1]} \right) \,-\, X_{[t]}$
forms a Markov chain,
so
$I \! \left( X_{[t]};U_t\; \big| \; B_{[t]},U_{[t-1]} \right)=0$.
Thus,
\begin{align}
    I \! & \left( X_{[t]} ; U_t\; \big| \; U_{[t-1]} \right)
    \le I \! \left( X_{[t]},B_{[t]};U_t\; \big| \; U_{[t-1]} \right) \\
    &\hspace{0.6cm}= I \! \left( B_{[t]};U_t\; \big| \; U_{[t-1]} \right)
    + I \! \left( X_{[t]};U_t\; \big| \; B_{[t]},U_{[t-1]} \right) \\
    &\hspace{0.6cm}= I \! \left( B_{[t]};U_t\; \big| \; U_{[t-1]} \right) \label{eq:cond-data-proc}.
\end{align}

Fix any encoding-and-control policy with
$\eqref{eq:cost} \le D$.
Its induced state-to-control law is feasible for~\eqref{eq:def-FnD},
hence
\begin{align*}
    nF_n(D)
    &\le I(X_{[n]}\to U_{[n]}) \\
    &\le I(B_{[n]}\to U_{[n]}) \le I(B_{[n]};U_{[n]}) \le H(B_{[n]}) \\
    &= \sum_{t=1}^n H(B_t\mid B_{[t-1]}) \le \sum_{t=1}^n \EE[\ell(B_t)],
\end{align*}
where
the second inequality follows by summing~\eqref{eq:cond-data-proc},
the third inequality is because
mutual information upper-bounds directed information~\cite[Thm.~1]{Massey1990},
and the last inequality uses the prefix-freeness of each
$\mc B_t(b_{[t-1]})$~\cite[Thm.~5.3.1]{CoverThomas2006}.
Dividing by $n$
and taking the infimum over such policies proves
\eqref{eq:conv}; if none exist, ${R_n(D)=\infty}$ and the bound is trivial.
\hfill$\square$

\subsection{Proof of Theorem~\ref{thm:ach-sfrl}} \label{subsec:proof-thm:ach-sfrl}

We consider a finite cost threshold $D\!\ge\!0$.
If ${F_n(D) \!=\! \infty}$, the bound is trivial; hence assume $F_n(D) \!<\! \infty$.
Fix a rate-cost pair $(R,D)$ such that
$R$ is strictly greater than ${F_n(D) \!+\! \log\bigl(F_n(D)+3.4\bigr) \!+\! 2 \!+\! \tfrac1n}$.
Denote the difference by $\gamma \!>\! 0$:
\begin{equation} \label{eq:ach-sketch-0}
    R = F_n(D) + \log\bigl(F_n(D)+3.4\bigr) + 2 + \frac1n + \gamma .
\end{equation}
Pick a sufficiently small $\varepsilon>0$ that satisfies
\begin{equation} \label{eq:eps-choice}
    2\varepsilon + \log\bigl(F_n(D)+\varepsilon+3.4\bigr) - \log\bigl(F_n(D)+3.4\bigr) \, \le \, \gamma.
\end{equation}

We pick a near-optimal action law.
By definition of $F_n(D)$, there exists a kernel $P_{U_{[n]}\parallel X_{[n]}}$ such that
\begin{equation} \label{eq:near-opt}
    \frac1n \sum_{t=1}^n \EE[c(X_t, U_t)] \le D, \quad
        \frac1n I \! \left( X_{[n]} \to U_{[n]} \right) \le F_n(D) + \varepsilon .
\end{equation}
The rest of the proof consists of three steps:
\begin{enumerate}
    \renewcommand{\labelenumi}{\arabic{enumi})}

    \item realizing
    the
    kernel $P_{U_{[n]}\parallel X_{[n]}}$ stage by stage via conditional SFRL;

    \item compressing
    the resulting auxiliary randomness to a binary time-sharing variable;

    \item constructing
    the encoding-and-control scheme and bounding its rate.
\end{enumerate}

\smallskip
\emph{Step 1: Causal SFRL representation.}
The first main challenge is that the action sequence must be generated causally.
For each ${t \in [n]}$,
Lemma~\ref{lem:cond-sfrl}
with
${(X,Y,U) = \left( X_{[t]},U_t,U_{[t-1]} \right)}$
yields a random variable
$Z_t \perp \left( X_{[t]},U_{[t-1]} \right)$
and
$g_t$ such that
\begin{equation} \label{eq:sfrl-func}
    U_t = g_t \! \left( X_{[t]},U_{[t-1]},Z_t \right) ,
\end{equation}
and, furthermore,
the (discrete) conditional entropy $H \! \left(U_t\; \big| \; U_{[t-1]},Z_t \right)$ is well-defined and finite, and satisfies
\begin{equation} \label{eq:sfrl-ent}
    \begin{aligned}
        H \! \left(U_t\; \big| \; U_{[t-1]},Z_t \right)
        &\le
        I \! \left( X_{[t]};U_t\; \big| \; U_{[t-1]} \right) \\
        & \hspace{-1.0cm} +
        \log \Bigl( I \! \left( X_{[t]};U_t\; \big| \; U_{[t-1]} \right)+3.4 \Bigr)
        +
        1.
    \end{aligned}
\end{equation}

Let $\delta_a$ denote the point mass at $a$.
We denote the deterministic action kernel that corresponds to~\eqref{eq:sfrl-func} by
\begin{equation}
    \widetilde P_{U_t \mid X_{[t]},U_{[t-1]},Z_t} \! \left( \cdot \mid x_{[t]},u_{[t-1]},z_t \right)
    \coloneq
    \delta_{g_t(x_{[t]},u_{[t-1]},z_t)} (\cdot)
\end{equation}
and we realize the representations jointly by
\begin{equation}\label{eq:sfrl-joint}
    P_{X_{[n]},U_{[n]},Z_{[n]}}
    =
    P_{X_{[n]}\parallel U_{[n-1]}}
    \prod_{t=1}^{n}
    \!\left(
        P_{Z_t} \cdot
        \widetilde P_{U_t \mid X_{[t]},U_{[t-1]},Z_t}
    \right)
\end{equation}
where $P_{Z_t}$ is the $t$\textsuperscript{th} SFRL auxiliary law,
and the factors are composed in the order $X_1,Z_1,U_1,\ldots,X_n,Z_n,U_n$.
This gives ${P_{Z_{[n]}}=\prod_{t=1}^{n}P_{Z_t}}$ and
${Z_t\perp(X_{[t]},U_{[t-1]},Z_{[t-1]})}$.
If $(X_{[t]},U_{[t-1]})$ has its selected distribution, this independence
and~\eqref{eq:sfrl-func} reproduce the joint distribution of
$(X_{[t]},U_{[t]},Z_t)$ supplied by the $t$\textsuperscript{th} SFRL application.
The prescribed state kernel then gives the selected distribution of the
next history.
Starting from $X_1$, induction therefore recovers the selected joint
state-action law and preserves~\eqref{eq:sfrl-ent} for every $t$.

By~\eqref{eq:sfrl-joint},
$P_{X_t\mid X_{[t-1]},U_{[t-1]},Z_{[n]}=z_{[n]}} = P_{X_t\mid X_{[t-1]},U_{[t-1]}}$
for $P_{Z_{[n]}}$-almost every $z_{[n]}$,
for $t = 1,\ldots,n$.
Thus, conditioning on $Z_{[n]}=z_{[n]}$ leaves the prescribed
initial-state law and state kernels unchanged,
while the actions satisfy $U_t=g_t(X_{[t]},U_{[t-1]},z_t)$.
This is a deterministic causal policy: $z_t$ is a fixed constant and no
future state is used.
Conditional discreteness in Lemma~\ref{lem:cond-sfrl} is preserved
almost surely by further conditioning on the remaining auxiliary
variables.
Induction over $t$ therefore implies that $U_{[n]}$ has countable support
given $Z_{[n]}=z_{[n]}$ for $P_{Z_{[n]}}$-almost every $z_{[n]}$.

By the chain rule and the fact that conditioning reduces entropy, we have
\begin{align}
    H \! \left(U_{[n]}\; \big| \; Z_{[n]} \right)
    \, &= \, \sum_{t=1}^n H \! \left(U_t\; \big| \; U_{[t-1]},Z_{[n]} \right) \notag \\
    \, &\le \, \sum_{t=1}^n H \! \left(U_t\; \big| \; U_{[t-1]},Z_t \right) . \label{eq:chain-entropy}
\end{align}
\eqref{eq:chain-entropy} and~\eqref{eq:sfrl-ent} yield
\begin{equation} \label{eq:HAnZ}
    \begin{aligned}
        H \! \left(U_{[n]}\; \big| \; Z_{[n]} \right)
        &\le
        I \! \left( X_{[n]}\to U_{[n]} \right) \\
        &\hspace{-0.75cm}+
        \sum_{t=1}^n \log \Bigl( I \! \left( X_{[t]};U_t\; \big| \; U_{[t-1]} \right) + 3.4 \Bigr)
        +
        n .
    \end{aligned}
\end{equation}
Applying Jensen's inequality to the concave function $\log(\cdot)$ yields
\begin{equation} \label{eq:entropy-rate-bound}
    \begin{aligned}
        \frac1n H \! \left(U_{[n]}\; \big| \; Z_{[n]} \right)
        &\le
        \frac1n I \! \left( X_{[n]}\to U_{[n]} \right) \\
        & \hspace{-0.5cm} +
        \log \! \left( \frac1n I \! \left( X_{[n]}\to U_{[n]} \right)+3.4 \right)
        +
        1.
    \end{aligned}
\end{equation}

\emph{Step 2: Reduction to binary time sharing.}
Our next goal is to reduce $Z_{[n]}$ to a binary time-sharing variable so that we can replace $Z_{[n]}$ by a selector with only one bit of entropy overhead.
This is needed because~\eqref{eq:entropy-rate-bound} controls
$H \! \left(U_{[n]}\mid Z_{[n]}\right)$, whereas the eventual coding
argument requires a bound on $H \bigl(\widetilde U_{[n]}\bigr)$
for the new action process; a binary variable $Q$ will let us write
${H \bigl(\widetilde U_{[n]}\bigr)
\le H \bigl(\widetilde U_{[n]}\mid Q\bigr)+H(Q)
\le H \bigl(\widetilde U_{[n]}\mid Q\bigr)+1}$.

Use the version of the conditional law given $Z_{[n]}=z_{[n]}$
obtained by fixing $z_{[n]}$ in the preceding sequential construction.
For each realization $z_{[n]}$
for which the following quantities are finite,
define
    \begin{align}
        d \! \left( z_{[n]} \right) &\coloneq \frac1n\sum_{t=1}^n \EE \! \left[ c(X_t,U_t) \, \big| \, Z_{[n]} = z_{[n]} \right], \label{eq:defn-d-r-2} \\
        \qquad r \! \left( z_{[n]} \right) &\coloneq \frac1n H \! \left( U_{[n]} \, \big| \, Z_{[n]} = z_{[n]} \right) \label{eq:defn-d-r} .
    \end{align}
Note that under the original distribution of $Z_{[n]}$, we have
\begin{equation} \label{eq:law-of-Z}
    \EE \! \left[ d \! \left( Z_{[n]} \right) \right]\le D , \quad \EE \! \left[ r \! \left( Z_{[n]} \right) \right]=\frac1n H \! \left(U_{[n]}\; \big| \; Z_{[n]} \right) .
\end{equation}
By~\eqref{eq:near-opt} and~\eqref{eq:entropy-rate-bound},
$d(Z_{[n]})$ and $r(Z_{[n]})$ are nonnegative with finite expectations,
and hence are finite for $P_{Z_{[n]}}$-almost every $z_{[n]}$.
Let $\mc Z_{[n]}$ denote a measurable set of such realizations with ${P_{Z_{[n]}} (\mc Z_{[n]} ) = 1}$
on which the preceding causal and countable-support properties also hold.
The functions $r,d$ are measurable on this set.
Define the region $\mc R$ as
\begin{equation} \label{eq:defn-mc-R}
    \mc R \coloneq \left\{ \Bigl( r \! \left( z_{[n]} \right) , d \! \left( z_{[n]} \right) \Bigr) \colon \, z_{[n]}\in\mc Z_{[n]} \right\} \subset \RR^2 .
\end{equation}
The law of $Z_{[n]}$ induces a probability measure on $\mc R$ with barycenter $\left( \bar r , \bar d \right)$ given by
\begin{equation}
    \left( \bar r , \bar d \right) \coloneq \Bigl( \EE \! \left[ r \! \left( Z_{[n]} \right) \right] , \, \EE \! \left[ d \! \left( Z_{[n]} \right) \right] \Bigr) .
\end{equation}
We prove that
$\left( \bar r , \bar d \right) \in \oconv{\mc R}$,
where $\conv{\cdot}$ and $\oconv{\cdot}$ denote the convex hull and its closure, respectively.
The preceding integrability statements imply that $\bigl(r(Z_{[n]}),d(Z_{[n]})\bigr)$ is an
integrable $\RR^2$-valued random vector,
and $\bigl(r(Z_{[n]}),d(Z_{[n]})\bigr) \in \mc R \subseteq \oconv{\mc R}$ almost surely.
Suppose, for contradiction, that $(\bar r,\bar d) \notin \oconv{\mc R}$.
Since $\oconv{\mc R}$ is a nonempty closed convex subset of $\RR^2$,
strict separation of a point and a closed convex set~\cite[Ex.~2.20]{Boyd2004convex} implies that there exists
$v\in\RR^2$ and $\alpha\in\RR$ such that
$v^\top (\bar r,\bar d) > \alpha \ge v^\top y$ for every $y\in \oconv{\mc R}$.
Because $\bigl(r(Z_{[n]}),d(Z_{[n]})\bigr) \in \oconv{\mc R}$ almost surely,
taking expectations yields
\begin{equation}
    v^\top (\bar r,\bar d)=\EE \! \left[ v^\top \bigl(r(Z_{[n]}),d(Z_{[n]})\bigr) \right] \le \alpha,
\end{equation}
a contradiction. Therefore
$(\bar r,\bar d)\in\oconv{\mc R}$.
Recall by~\eqref{eq:law-of-Z} that $\bar d\le D$.

\smallskip
We now show that
\begin{equation} \label{eq:claim}
    \forall \, \delta>0 \ \exists \, (r_\delta,d_\delta) \in \conv{\mc R} \ \text{s.t.} \ d_\delta\le D , \, r_\delta\le \bar r+\delta .
\end{equation}
Fix $\delta>0$. Since
$(\bar r,\bar d)\in\oconv{\mc R}$, choose a sequence
${(a_k,b_k)\in\conv{\mc R}}$ such that
$(a_k,b_k)\longrightarrow(\bar r,\bar d)$.
If $\bar d<D$, then, for all sufficiently large $k$,
$b_k\le D$ and $a_k\le\bar r+\delta$.
Taking $(r_\delta,d_\delta)=(a_k,b_k)$
proves~\eqref{eq:claim} in this case.

Now suppose $\bar d=D$ and there exists
$(r_0,d_0)\in\mc R$ with $d_0<D$.
Define
\begin{equation*}
    \theta_k\coloneq
    \begin{cases}
        \displaystyle\frac{b_k-D}{b_k-d_0}, & b_k>D,\\[1ex]
        0, & b_k\le D,
    \end{cases}
\end{equation*}
and form the corrected pair
\begin{equation*}
    (\widetilde a_k,\widetilde b_k)
    \coloneq
    (1-\theta_k)(a_k,b_k)+\theta_k(r_0,d_0).
\end{equation*}
Since $0\le\theta_k<1$, this pair belongs to
$\conv{\mc R}$. Its cost coordinate satisfies
\begin{equation*}
    \widetilde b_k
    =b_k-\theta_k(b_k-d_0)
    =\min\{b_k,D\}\le D.
\end{equation*}
Moreover, because $b_k\to D$,
\begin{equation*}
    0\le\theta_k
    \le\frac{|b_k-D|}{D-d_0}
    \longrightarrow0.
\end{equation*}
Since $a_k\to\bar r$, the sequence $|r_0-a_k|$ is
bounded, and hence
\begin{equation*}
    |\widetilde a_k-\bar r|
    \le |a_k-\bar r|+\theta_k|r_0-a_k|
    \longrightarrow0.
\end{equation*}
Thus, for sufficiently large $k$, taking
$(r_\delta,d_\delta)
=(\widetilde a_k,\widetilde b_k)$ gives
${d_\delta\le D}$ and
$r_\delta\le\bar r+\delta$.

Finally, suppose $\bar d=D$ and no point of $\mc R$
has cost below $D$. Then $d(Z_{[n]})\ge D$ almost surely,
and
\begin{equation*}
    \EE[d(Z_{[n]})-D]=\bar d-D=0.
\end{equation*}
Consequently, $d(Z_{[n]})=D$ almost surely.
Also, $\EE[r(Z_{[n]})]=\bar r$ implies
\begin{equation*}
    \PP\!\left(r(Z_{[n]})\le\bar r+\delta\right)>0;
\end{equation*}
otherwise, the expectation would be at least
$\bar r+\delta$, a contradiction.
We may therefore choose a realization
$z_{[n]} \in \mc Z_{[n]}$
with
\begin{equation*}
    d(z_{[n]})=D,
    \quad
    r(z_{[n]})\le\bar r+\delta.
\end{equation*}
Taking
$(r_\delta,d_\delta)
=(r(z_{[n]}),d(z_{[n]}))\in\mc R\subseteq\conv{\mc R}$
proves~\eqref{eq:claim} in the remaining case.
No compactness or closedness of $\mc R$ is required.

Apply~\eqref{eq:claim} with $\delta = \varepsilon$, where $\varepsilon$ is the value fixed in~\eqref{eq:eps-choice}, and denote the resulting point by $(r_\varepsilon,d_\varepsilon)$.
We now reduce $Z_{[n]}$ to a binary time-sharing variable.
Following the reasoning in
\cite[Sec.\ IV-A]{LiElGamal2018SFRL}, because we are in $\RR^2$,
by Carath\'eodory's theorem, any ${(r_\varepsilon,d_\varepsilon) \in \conv{\mc R}}$ lies in the convex hull of
(at most)
three points of $\mc R$;
denote
this (possibly degenerate)
triangle by $T$.
Consider the rectangle
\begin{equation} \label{eq:defn-rectangle}
    S \coloneq (-\infty,r_\varepsilon] \times (-\infty,D] .
\end{equation}
We know $T\cap S$ is nonempty since $(r_\varepsilon,d_\varepsilon) \in S$.
The set $T$ is compact as the convex hull of finitely many points,
and $S$ is closed; hence $T\cap S$ is compact.
Choose $q$ minimizing the sum of its coordinates over $T\cap S$.
If $T$ is nondegenerate, $q$ lies on an edge of $T$: otherwise a small
move in direction $(-1,-1)$ remains in $T\cap S$ and decreases that sum.
If $T$ is degenerate, it is already a segment or a point.
Thus, there exist vertices $q_0, q_1$ of $T$ (possibly equal), and a constant $\lambda\in[0,1]$ such that
\begin{equation}
    q = (1-\lambda) q_0 + \lambda q_1 .
\end{equation}
Since the vertices of $T$ belong to $\mc R$, we have $q_0,q_1\in\mc R$.
Denote by
$z_{[n]}(0),z_{[n]}(1)$
the points corresponding to $q_0,q_1$ in the sense of~\eqref{eq:defn-mc-R}.
With $Q\sim\mathrm{Bern}(\lambda)$,
define the interpolation
\begin{equation} \label{eq:defn-Zn-Q}
    z_{[n]}(Q) \coloneq
    \begin{cases}
        z_{[n]}(0), & \text{if } Q=0, \\
        z_{[n]}(1), & \text{if } Q=1,
    \end{cases}
\end{equation}
which corresponds to the point $q \in S$, and hence, by~\eqref{eq:defn-rectangle}, satisfies the inequalities
\begin{equation} \label{eq:binary-Q}
    \EE \! \left[ d \! \left( z_{[n]}(Q) \right) \right] \, \le \, D,
    \quad
    \EE \! \left[ r \! \left( z_{[n]}(Q) \right) \right] \, \le \, \bar r + \varepsilon .
\end{equation}
This gives us the desired binary time-sharing variable.
Take $Q$ independent of the randomness driving the plant, and
generate
$(\widetilde X_{[n]},\widetilde U_{[n]})$ sequentially using the same
initial-state law and prescribed state kernels as the original process,
with the causal law
${\widetilde U_t=g_t \bigl( \widetilde X_{[t]},\widetilde U_{[t-1]},z_t(Q) \bigr)}$.
By the sequential construction, conditional on $Q=q$
the joint law of
$\bigl( \widetilde X_{[n]},\widetilde U_{[n]} \bigr)$ coincides with the original
conditional law of $(X_{[n]},U_{[n]})$ given
$Z_{[n]}=z_{[n]}(q)$: both use the same deterministic causal policy,
initial-state law, and state kernels.
The unconditional joint law is the corresponding binary mixture.
Thus, $\widetilde X_{[n]}$ denotes the plant state sequence under the
new policy, whose distribution need not equal that of $X_{[n]}$.
The new law need not reproduce the selected causal kernel, but it
satisfies the cost constraint by~\eqref{eq:binary-Q}.
Since both selected action-sequence laws have countable support,
$\widetilde U_{[n]}$ also has countable support.

\medskip
Next,
we upper-bound
conditional entropy.
\begin{align}
    H\!\left(\widetilde U_{[n]}\mid Q\right)
    &= \sum_{q=0}^1 P_Q(q)
       H\!\left(U_{[n]}\mid Z_{[n]}=z_{[n]}(q)\right)
       \label{eq:replace-Q-by-Z} \\
    &= n\,\EE\!\left[r\!\left(z_{[n]}(Q)\right)\right].
       \label{eq:last}
\end{align}
Equality~\eqref{eq:replace-Q-by-Z} follows from the equality of the
conditional action laws just established, and~\eqref{eq:last}
follows from~\eqref{eq:defn-d-r}.
Equations~\eqref{eq:binary-Q} and~\eqref{eq:law-of-Z} therefore yield
\begin{equation} \label{eq:HQ}
    H\!\left(\widetilde U_{[n]}\mid Q\right)
    \le H\!\left(U_{[n]}\mid Z_{[n]}\right)+n\varepsilon.
\end{equation}

\smallskip
\emph{Step 3: Causal coding of the resulting actions.}
Lastly, we construct an encoding-and-control policy from $Q$.
Before transmission, the terminals agree on deterministic prefix-free
Shannon codes matched to the induced binary-mixture laws
$P_{\widetilde U_t\mid \widetilde U_{[t-1]}}$ (to $P_{\widetilde U_1}$ at $t=1$).
These are discrete laws by the countable-support argument above;
$Q$ is averaged out according to $P_{Q\mid \widetilde U_{[t-1]}}$.
The encoder privately samples $Q\sim\mathrm{Bern}(\lambda)$ before $t=1$,
and uses the same realization throughout the horizon.
At each stage, both terminals recover $\widetilde U_{[t-1]}$ from $B_{[t-1]}$
and hence select the same code.
The encoder computes
\begin{equation}
    \widetilde U_t = g_t\!\left(\widetilde X_{[t]},\widetilde U_{[t-1]},z_t(Q)\right)
\end{equation}
and sends its codeword $B_t\in\mc B_t(B_{[t-1]})$.
The controller decodes and applies $\widetilde U_t$
without knowing $Q$;
no shared randomness is required.
Marginalizing the encoder's private $Q$ gives encoding kernels of
the form~\eqref{eq:encoding-kernel}.

For each $t$, the conditional Shannon code satisfies~\cite[Ex.\ 5.28]{CoverThomas2006},~\cite[Thm.\ 2]{Tanaka2016RatePrefixFreeLQG}
\begin{equation} \label{eq:thm1-conc-1}
    \EE[\ell(B_t)] \le H \! \left(\widetilde U_t\; \big| \; \widetilde U_{[t-1]} \right)+1.
\end{equation}
Summing for $t=1,\ldots,n$, we obtain
\begin{align}
    \sum_{t=1}^n \EE[\ell(B_t)]
    &\le \sum_{t=1}^n H \! \left(\widetilde U_t \,\big|\, \widetilde U_{[t-1]} \right) + n
    = H \! \left(\widetilde U_{[n]} \right) + n \\
    &\le H \! \left(\widetilde U_{[n]}\; \big| \; Q \right)
    + H(Q)
    + n 
    \le nR
    \label{eq:bar-Rn-le-R}
\end{align}
where
the last inequality
follows from~\eqref{eq:HQ},
\eqref{eq:entropy-rate-bound},~\eqref{eq:near-opt}, and the choice
\eqref{eq:eps-choice} together with~\eqref{eq:ach-sketch-0}.

Recall that,
by~\eqref{eq:binary-Q}, the new process satisfies
\begin{equation*}
    \frac1n\sum_{t=1}^n \EE[c(\widetilde X_t,\widetilde U_t)]
    = \EE[d(z_{[n]}(Q))]\le D.
\end{equation*}
Together with~\eqref{eq:bar-Rn-le-R},
this implies that any rate-cost pair $(R,D)$
with ${R > F_n(D) + \log\bigl(F_n(D)+3.4\bigr) + 2 + \tfrac1n}$
is achievable at time horizon $n$,
which yields Theorem~\ref{thm:ach-sfrl}.
\hfill
$\square$

\subsection{Proof of Proposition~\ref{prop:stability}} \label{subsec:proof-stability}

For the converse, Proposition~\ref{prop:conv} gives
${R\ge R_n(D)\ge F_n(D)}$ for every sufficiently large $n$.
Taking the limsup yields ${R\ge\overline F(D)\ge F_{\mathrm{stab}}}$.

For achievability,
denote $g(x)=x+\log(x+3.4)+2$.
Since $g$ is increasing and continuous on $[0,\infty)$,
\eqref{eq:stab-ach} and the definition of $F_{\mathrm{stab}}$
allow us to choose $D<\infty$ with $g(\overline F(D))<R$.
The definition of the limsup then gives
\begin{equation*}
    g \bigl( F_n(D) \bigr)+\tfrac1n<R
\end{equation*}
for all sufficiently large $n$.
Theorem~\ref{thm:ach-sfrl} therefore implies $R_n(D)<R$,
so a horizon-$n$ policy with average rate at most $R$ and
average cost at most $D$ exists for each such $n$.


\subsection{Proof of Proposition~\ref{prop:lqg-stability}} \label{subsec:proof-vector-lqg}

\subsubsection{Proof of~\eqref{eq:vector-Fstab}} \label{subsec:app-A}

We separately prove the two inequalities
$F_{\mathrm{stab}}\ge h(A)$ and $F_{\mathrm{stab}}\le h(A)$.
By~\cite[Def.~1 and Thm.~3]{KostinaHassibi2019RateCostControl},
we have ${\overline F(D)\ge h(A)}$ for all sufficiently large $D$.
Since $\overline F(D)$ is nonincreasing in $D$,
the same bound holds for every finite $D\ge0$.
Taking the infimum over $D$ gives
$F_{\mathrm{stab}}\ge h(A)$.

For the reverse inequality,
for every $\varepsilon>0$,
\cite[Thm.~2 and Cor.~1]{TanakaEsfahaniMitter2018_LQGMinDI}
provides an infinite-horizon causal state-to-control policy with finite
limsup average quadratic cost and directed-information rate
at most $h(A)+\varepsilon$.
The use of $X_{t+1}$ in~\cite{TanakaEsfahaniMitter2018_LQGMinDI}
rather than $X_t$ in the state cost does not affect finiteness:
our $n$-stage cost sum is at most the reference's sum plus $\EE[X_1^TQX_1]$.
We choose $D$ strictly above the policy's limsup average cost.
Its truncations are feasible for $F_n(D)$ for all sufficiently
large $n$, so we obtain
\begin{equation}
    F_{\mathrm{stab}}
    \le \overline F(D)
    \le h(A)+\varepsilon.
\end{equation}
Letting $\varepsilon\downarrow0$
proves~\eqref{eq:vector-Fstab}.
The average-cost converse follows from
Proposition~\ref{prop:stability}.

\subsubsection{State-second-moment converse and achievability} \label{subsec:app-B}

For the
converse, consider the same plant
but with the auxiliary stage cost $c_0(x,u)=\|x\|^2$.
Let $R_{n,0}(C)$ and $F_{n,0}(C)$ denote the corresponding
auxiliary operational rate-cost function and directed-information
benchmark, respectively.

Suppose that, for every sufficiently large $n$, there is a
horizon-$n$ encoding-and-control policy with average
rate at most
$R$ and
$\frac{1}{n}\sum_{t=1}^n \EE\!\left[\|X_t\|^2\right]\le C$,
where $R,C<\infty$ are independent of $n$.
Under the auxiliary cost $c_0$, this policy is feasible at cost
level $C$, and hence
$R_{n,0}(C)\le R$.
Applying Proposition~\ref{prop:conv} to the same plant with
the auxiliary cost $c_0$ gives
$R_{n,0}(C)\ge F_{n,0}(C)$.
Therefore,
$R\ge F_{n,0}(C)$
for every sufficiently large $n$, and consequently
\begin{equation}
    R
    \ge \limsup_{n\to\infty}F_{n,0}(C).
\end{equation}
The
vector
LQG converse
\cite[Def.~1 and Thm.~3]{KostinaHassibi2019RateCostControl}
applied
with ${(Q_0, M_0, S_n)=(I_k, 0, I_k)}$,
gives,
by the
monotonicity in the cost threshold,
\begin{equation}
    \limsup_{n\to\infty}F_{n,0}(C)\ge h(A)
\end{equation}
for every finite $C$.
Thus,
$R\ge h(A)$,
as desired.

\smallskip
For achievability, substituting~\eqref{eq:vector-Fstab} into the
achievability statement of Proposition~\ref{prop:stability} shows that
every rate satisfying~\eqref{eq:vector-state-ach} admits, for
some fixed $D<\infty$ and every sufficiently large $n$, a
horizon-$n$ encoding-and-control policy with average
rate at most
$R$ and
$\frac{1}{n}\sum_{t=1}^n \EE\!\left[c(X_t,U_t)\right]\le D$.
Since
\begin{equation}
    c(x,u)
    = x^TQx+u^TMu
    \ge x^TQx
    \ge \lambda_{\min}(Q)\|x\|^2,
\end{equation}
these policies also satisfy
\begin{equation}
    \frac{1}{n}\sum_{t=1}^n
    \EE\!\left[\|X_t\|^2\right]
    \le
    \frac{D}{\lambda_{\min}(Q)} ,
\end{equation}
for all sufficiently large $n$.
Since $D$ is independent of $n$,
both the average cost and the time-averaged state second moment
admit horizon-independent bounds.

\section{Conclusion} \label{sec:conclusion}

We studied rate-limited control for general (including nonlinear)
control systems over a finite horizon.
Our key result is that the operational rate-cost function $R_n(D)$ is governed by a directed-information minimization.
Specifically, by proving the converse and achievability bounds
in~\eqref{eq:Rn-Fn-0},
we established $F_n(D)$ in~\eqref{eq:def-FnD} as the fundamental information quantity for the operational rate-cost tradeoff, up to an additive logarithmic gap.

The infinite-horizon behavior of $R_n(D)$ and $F_n(D)$ is of interest for understanding steady-state operation and long-run average performance.
Constant-size characterizations of $\lim_{n\to\infty} F_n(D)$ are known
for LQG systems:
the infinite-horizon directed-information optimization admits a steady-state, finite-dimensional characterization through Riccati equations and semidefinite programming methods~\cite{TanakaEsfahaniMitter2018_LQGMinDI}.
A natural future direction is to identify broader families of control systems
whose 
infinite-horizon limit $\lim_{n\to\infty} F_n(D)$ admits a constant-size reduction---for example, nonlinear systems under suitable Markov and stationarity assumptions.

\bibliographystyle{IEEEtran}
\bibliography{references_abbr.bib}

\end{document}